\documentclass[submission,copyright,creativecommons]{eptcs}
\providecommand{\event}{GandALF 2018} 
\usepackage{underscore}           

\usepackage[final]{microtype}
\usepackage{graphicx}

\usepackage{paralist}
\usepackage{tikz}
\usetikzlibrary{positioning,arrows}
\usepackage{amsmath,amsfonts,amssymb,amsthm}
\usepackage{hhline}
\usepackage{multirow}

\usepackage{multicol}
\usepackage{mathtools}

\usepackage{float}
\usepackage{ltlfonts} %
\usepackage{xspace}
\usepackage{latexsym}

\newcommand{\EXPTIME}{{\sc Exptime}\xspace}

\newcommand{\PSPACE}{{\sc Pspace}\xspace}
\newcommand{\EXPSPACE}{{\sc Expspace}\xspace}

\newcommand{\Nat}{{\mathbb{N}}}
\newcommand{\RealP}{{\mathbb{R}_+}}

\DeclareMathAlphabet{\mathpzc}{OT1}{pzc}{m}{it}

\newcommand{\Prop}{\mathcal{P}}

\newcommand{\tpl}[1]{(#1)}
\newcommand{\Lang}{{\mathcal{L}}}
\newcommand{\TLang}{{\mathcal{L}_T}}
\newcommand{\TLangInf}{{\mathcal{L}^{\omega}_T}}

\newtheorem{definition}{Definition}
\newtheorem{theorem}{Theorem}
\newtheorem{proposition}{Proposition}
\newtheorem{lemma}{Lemma}

\newcommand{\MSO}{\text{\sffamily MSO}}
\newcommand{\TA}{\text{\sffamily TA}}
\newcommand{\PDA}{\text{\sffamily PDA}}

\newcommand{\VPA}{\text{\sffamily VPA}}
\newcommand{\VPTA}{\text{\sffamily VPTA}}

\newcommand{\VPL}{\text{\sffamily VPL}}

\newcommand{\ECA}{\text{\sffamily ECA}}
\newcommand{\CARET}{\text{\sffamily CaRet}}
\newcommand{\ECNA}{\text{\sffamily ECNA}}

\newcommand{\ECVPA}{\text{\sffamily ECVPA}}
\newcommand{\LTL}{\text{\sffamily LTL}}
\newcommand{\MTL}{\text{\sffamily MTL}}
\newcommand{\NMTL}{\text{\sffamily NMTL}}
\newcommand{\MITL}{\text{\sffamily MITL}}
\newcommand{\MITLS}{\text{{\sffamily MITL}$_{(0,\infty)}$}}
\newcommand{\NMITLS}{\text{{\sffamily NMITL}$_{(0,\infty)}$}}
\newcommand{\ECTL}{\text{\sffamily EC\_TL}}
\newcommand{\ECNTL}{\text{\sffamily EC\_NTL}}

\newcommand{\Cl}{\textsf{Cl}}

\newcommand{\Next}{\LTLcircle}
\newcommand{\Prev}{\LTLcircleminus}
\newcommand{\Always}{\LTLsquare}
\newcommand{\Eventually}{\LTLdiamond}
\newcommand{\StrictAlways}{\LTLsquarehat}
\newcommand{\StrictEventually}{\LTLdiamondhat}
\newcommand{\StrictPastAlways}{\LTLsquareminushat}
\newcommand{\StrictPastEventually}{\LTLdiamondminushat}

\newcommand{\NextClock}{\rhd}
\newcommand{\PrevClock}{\lhd}
\newcommand{\Until}{\textsf{U}}
\newcommand{\Since}{\textsf{S}}
\newcommand{\StrictUntil}{\widehat{\textsf{U}}}
\newcommand{\StrictSince}{\widehat{\textsf{S}}}

\newcommand{\caller}{\mathsf{c}}
\newcommand{\Global}{\mathsf{g}}
\newcommand{\Scall}{\Sigma_{\mathit{call}}}
\newcommand{\call}{{\mathit{call}}}
\newcommand{\ret}{{\mathit{ret}}}
\newcommand{\intA}{{\mathit{int}}}
\newcommand{\Caller}{{\mathit{Caller}}}
\newcommand{\Sret}{\Sigma_{\mathit{ret}}}
\newcommand{\Sint}{\Sigma_{\mathit{int}}}
\newcommand{\Au}{\ensuremath{\mathcal{A}}}

\newcommand{\abs}{\mathsf{a}}
\newcommand{\SUCC}{\mathsf{succ}}
\newcommand{\NULL}{\mathsf{\vdash}}
\newcommand{\Pos}{{\mathit{Pos}}}
\newcommand{\val}{{\mathit{val}}}
\newcommand{\Const}{{\mathit{Const}}}

\newcommand{\dir}{\textit{dir}}
\newcommand{\Proj}{\textit{Proj}}
\newcommand{\NextPrev}{\textit{Next}}
\newcommand{\AbsNextPrev}{\textit{AbsNext}}

\newcommand{\MAP}{\textit{MAP}}

\newcommand{\INTS}{\ensuremath{\mathcal{I}_{(0,\infty)}}}

\newcommand{\Lab}{{\textit{Lab}}}
\newcommand{\Succ}{{\textit{succ}}}

\newcommand{\halt}{{\textit{halt}}}
\newcommand{\init}{{\textit{init}}}

\newcommand{\Inst}{\mathsf{Inst}}

\newcommand{\dec}{{\textit{dec}}}
\newcommand{\zero}{{\textit{zero}}}

\newcommand{\Fragm}{\mathcal{F}}

\newcommand{\DefORmini}{\ensuremath{\;\big|\;}}
\newcommand{\true}{\ensuremath{\top}}

\newcommand{\details}[1]{{}}%

\title{Timed Context-Free Temporal Logics\footnote{The work by Adriano Peron and Aniello Murano has been partially supported by the  GNCS project \emph{Formal methods for verification and synthesis of discrete and hybrid systems} and by Dept. project MODAL \emph{MOdel-Driven Analysis of Critical Industrial Systems.}}}
\author{Laura Bozzelli \qquad Aniello Murano \qquad Adriano Peron
\institute{University of Napoli ``Federico II'', Napoli, Italy}
}

\def\titlerunning{Timed Context-Free Temporal Logics}
\def\authorrunning{L. Bozzelli, A. Murano \& A. Peron}

\begin{document}

%
\maketitle              

\begin{abstract}
The paper is focused on temporal logics for the description of the behaviour of
real-time pushdown reactive systems. The paper is motivated to bridge tractable logics specialized for expressing separately dense-time real-time properties and context-free properties
by ensuring decidability and tractability in the combined setting. 
To this end we introduce two real-time linear temporal logics for specifying quantitative timing context-free requirements in a pointwise semantics setting: \emph{Event-Clock Nested Temporal Logic} (\ECNTL) and \emph{Nested Metric Temporal Logic} (\NMTL). The logic \ECNTL\ is an extension of both the 
logic \CARET\ (a context-free
extension of standard \LTL) and \emph{Event-Clock Temporal Logic}
(a tractable real-time logical framework related to the class of Event-Clock automata).  We prove that 
satisfiability  of \ECNTL\ and visibly model-checking of Visibly Pushdown Timed Automata (\VPTA) against \ECNTL\ are decidable and \EXPTIME-complete. The other proposed logic \NMTL\ is a context-free extension of standard  Metric Temporal Logic (\MTL). It is well known that satisfiability of future
$\MTL$ is undecidable when interpreted over infinite timed words but decidable over finite timed words. On the other hand, 
we show that by augmenting future \MTL\ with future context-free temporal operators, the satisfiability problem turns out to be undecidable also for finite timed words. On the positive side, we devise a meaningful and decidable 
 fragment of the logic \NMTL\
which is expressively equivalent  to \ECNTL\ and for which satisfiability and visibly model-checking of \VPTA\ are  \EXPTIME-complete.
\end{abstract}

\section{Introduction}

\emph{Model checking} is a well-established formal-method technique to automatically check for global correctness of reactive systems~\cite{Baier}.
In this setting, temporal logics provide a fundamental framework for the description of the
dynamic behavior of reactive systems.

In the last two decades,
model checking of pushdown automata (\PDA) has received a lot of attention~\cite{Wal96,CMM+03,AlurMadhu04,BMP10}. 
\PDA\ represent an infinite-state formalism suitable to model the control flow of typical sequential programs with
nested and recursive procedure calls. Although  the general problem of checking context-free properties of \PDA\ is undecidable, 
algorithmic solutions have been proposed for interesting subclasses of context-free requirements~\cite{AlurEM04,AlurMadhu04,CMM+03}. 
A relevant example is that of the linear temporal logic \CARET\ \cite{AlurEM04}, a context-free
extension of standard \LTL.  \CARET\ formulas are interpreted on  words over a  \emph{pushdown alphabet}  which is partitioned into three
disjoint sets of calls, returns, and internal symbols. A call
 denotes invocation of a procedure (i.e.
 a push stack-operation)
and the \emph{matching} return (if any) along a given word denotes
the exit from this procedure (corresponding to a pop
stack-operation). \CARET\ allows to specify \LTL\ requirements over two kinds of \emph{non-regular}
patterns on input words: \emph{abstract paths} and \emph{caller paths}. An abstract path captures the local computation
within a procedure with the removal of subcomputations corresponding
 to nested procedure calls, while a caller path represents the call-stack content  at a given position of the input.
An automata theoretic generalization of \CARET\ is the class of (nondeterministic)
 \emph{Visibly Pushdown Automata} (\VPA)~\cite{AlurMadhu04}, a subclass 
of \PDA\ where the input symbols over a  pushdown alphabet  control the admissible operations
on the stack. \VPA\ push onto the stack only when a call is
 read, pops the stack only at returns, and do not use the stack on
 reading internal symbols.
 This restriction makes the class of
resulting languages  (\emph{visibly pushdown languages} or \VPL)
very similar in tractability and robustness to the less expressive class of regular languages~\cite{AlurMadhu04}. 
 In fact,
\VPL\ are closed under Boolean operations, and language inclusion, which is undecidable for context-free languages, is \EXPTIME-complete for \VPL.\vspace{0.1cm}

\noindent  \textbf{Real-time pushdown model-checking.} Recently, many works~\cite{AbdullaAS12,BenerecettiMP10,BenerecettiP16,BouajjaniER94,ClementeL15,EmmiM06,TrivediW10}  have investigated real-time extensions of \PDA\  by combining \PDA\ with
 \emph{Timed Automata} (\TA)~\cite{AlurD94}, a model widely used
to represent real-time systems. \TA\ are finite automata augmented with a finite set of real-valued clocks, which operate over words where each symbol is paired with a real-valued timestamp (\emph{timed words}).
All the clocks progress at the same speed and can
be reset by transitions (thus, each clock
keeps track of the elapsed time since the last reset). 
 The emptiness problem for \TA\ is decidable and \PSPACE-complete~\cite{AlurD94}.
However, since in \TA, clocks can be reset nondeterministically and independently of each other, the resulting class of timed languages is not closed under complement and, moreover,
language inclusion  is undecidable~\cite{AlurD94}. As a consequence, the
general verification problem (i.e., language inclusion) of formalisms combining unrestricted \TA\ with robust subclasses of \PDA\ such as \VPA , i.e.
\emph{Visibly Pushdown Timed Automata} (\VPTA), is undecidable as well.
In fact,  checking language inclusion for \VPTA\  is undecidable even in the restricted case of specifications using at most one clock~\cite{EmmiM06}.
More robust approaches \cite{TangO09,BhaveDKPT16,BMP18}, although less expressive, are based on formalisms combining \VPA\ and \emph{Event-clock automata} (\ECA)~\cite{AlurFH99} such as the recently introduced class
of \emph{Event-Clock Nested Automata} (\ECNA)~\cite{BMP18}.
\ECA\ \cite{AlurFH99} are a well-known determinizable subclass of \TA\ where the explicit reset of clocks is disallowed. In \ECA, clocks have a predefined association with the input alphabet symbols  and their values refer
to the time distances from previous and next occurrences of input symbols. \ECNA\ \cite{BMP18} combine \ECA\ and \VPA\ by providing an explicit mechanism to relate the use of a stack with that of event clocks.
In particular, \ECNA\ retain the 
closure and decidability properties of \ECA\ and \VPA\ being closed under Boolean operations and having a decidable (specifically, \EXPTIME-complete) language-inclusion problem, and are strictly more expressive
than other formalisms combining \ECA\ and \VPA\ \cite{TangO09,BhaveDKPT16} such as the class  of \emph{Event-Clock Visibly Pushdown Automata} (\ECVPA)~\cite{TangO09}. In~\cite{BhaveDKPT16} a logical characterization of the class
of \ECVPA\ is provided by means of a non-elementarily decidable extension of standard \MSO\ over words.\vspace{0.1cm}

\noindent  \textbf{Our contribution.} In this paper, we introduce two real-time linear temporal logics, called \emph{Event-Clock Nested Temporal Logic} (\ECNTL) and \emph{Nested Metric Temporal Logic} (\NMTL)
for specifying quantitative timing context-free requirements in a pointwise semantics setting (models of formulas are timed words).
The logic \ECNTL\ is an extension of \emph{Event-Clock Temporal Logic} (\ECTL)~\cite{RaskinS99}, the latter being a known decidable and tractable real-time logical framework related to the class of Event-clock automata.
\ECTL\ extends \LTL\ + past with timed temporal modalities which specify time constraints on the distances from the previous or next timestamp where a given subformula
holds. The novel logic \ECNTL\ is an extension of both \ECTL\ and \CARET\ by means of non-regular versions of the timed modalities of \ECTL\ which allow to refer to abstract and caller paths.
We address expressiveness and complexity issues for the logic \ECNTL. In particular, we establish that satisfiability of \ECNTL\ and visibly model-checking of \VPTA\ against \ECNTL\ are decidable and \EXPTIME-complete.
The key step in the proposed decision procedures is a translation
of \ECNTL\ into  \ECNA\ accepting suitable encodings of the models of the given formula.

The second logic we introduce, namely \NMTL, is a  context-free extension of standard  Metric Temporal Logic (\MTL). This extension is obtained by adding
to \MTL\  timed versions of the caller and abstract temporal modalities of \CARET. In the considered pointwise-semantics settings, it is well known that
satisfiability of future \MTL\ is undecidable when  interpreted over infinite timed words~\cite{OuaknineW06}, and decidable~\cite{OuaknineW07} over finite timed words. We show that over finite timed words, the adding of the future abstract timed modalities to future \MTL\   makes the satisfiability problem undecidable. On the other hand, we show that the fragment
\NMITLS\ of \NMTL\ (the \NMTL\ counterpart of the well-known tractable fragment \MITLS\ \cite{AlurFH96} of \MTL) has the same expressiveness
as the logic \ECNTL\ and the related satisfiability and visibly model-checking problems are \EXPTIME-complete. The overall picture of decidability results is reported in table~\ref{results} (new results in red).


\begin{table}[tb]
\begin{center}
	\caption{Decidability results.}\label{results}
		\begin{tabular}{ cclc }
			\hline
 Logic & Satisfiability &  Visibly model checking\\
			\hline
			
			
		{\color{red}	\ECNTL\ } & 	{\color{red} \EXPTIME-complete} & 	{\color{red} \EXPTIME-complete} \\
			
			{\color{red}	\NMITLS}\ & 	{\color{red} \EXPTIME-complete} & 	{\color{red} \EXPTIME-complete} \\
			
			future \MTL\ fin. & Decidable & \\
			
			future \MTL\ infin. & Undecidable & \\
			
			{\color{red} future \NMTL\ fin.} & 	{\color{red} Undecidable} &\\
			
			\hline
	\end{tabular}
\end{center}
\end{table}


Due to space limitations, some proofs are omitted: we refer for complete proofs to the complete version of the paper in
\cite{compGandalf18}.

\section{Preliminaries}\label{sec:backgr}

In the following, $\Nat$ denotes the set of natural numbers and $\RealP$ the set of non-negative real numbers.
Let $w$ be a finite or infinite word over some alphabet. By $|w|$ we denote the length of $w$ (we write $|w|=\infty$ if $w$ is infinite). For all  $i,j\in\Nat $, with $i\leq j <|w|$, $w_i$ is
$i$-th letter of $w$, while $w[i,j]$ is the finite subword
 $w_i\cdots w_j$.

A  \emph{timed word} $w$ over a finite alphabet $\Sigma$ is
a  word $w=(a_0,\tau_0) (a_1,\tau_1),\ldots$ over $\Sigma\times \RealP$ 
($\tau_i$ is the time at which $a_i$ occurs) such that the sequence $\tau= \tau_0,\tau_1,\ldots$ of timestamps  satisfies: (1) $\tau_{i-1}\leq \tau_{i}$ for all $0<i<|w|$ (monotonicity), and (2) if $w$ is infinite, then for all $t\in\RealP$, $\tau_i\geq t$ for some $i\geq 0$
(divergence). The timed word $w$ is also denoted by the pair $(\sigma,\tau)$, where $\sigma$ is the untimed word $a_0 a_1\ldots$.
A \emph{timed language} (resp., \emph{$\omega$-timed language}) over $\Sigma$ is a set of finite  (resp., infinite) timed words over $\Sigma$.\vspace{0.2cm}

\noindent \textbf{Pushdown alphabets, abstract paths, and caller paths.}
A \emph{pushdown alphabet} is
a finite alphabet $\Sigma=\Scall\cup\Sret\cup\Sint$ which is partitioned into a set $\Scall$ of \emph{calls}, a set
$\Sret$ of \emph{returns}, and a set $\Sint$ of \emph{internal
  actions}. The pushdown alphabet $\Sigma$  induces a nested hierarchical structure in a given word over $\Sigma$ obtained by associating
to each call the corresponding matching return (if any) in a well-nested manner. Formally, the set  of \emph{well-matched words} is the set of finite words $\sigma_w$ over
$\Sigma$ inductively defined as follows:
\[
\sigma_w:= \varepsilon \DefORmini
a\cdot \sigma_w \DefORmini
c\cdot \sigma_w \cdot r \cdot \sigma_w
\]
where $\varepsilon$ is the empty word, $a\in\Sint$, $c\in \Scall$, and $r\in\Sret$.

Fix a  word $\sigma$ over $\Sigma$. For a call position $i$ of $\sigma$,
if there is $j>i$ such that $j$ is a return position of $\sigma$ and
$\sigma[i+1,j-1]$ is a well-matched word (note that $j$ is
uniquely determined if it exists), we say that $j$ is the
\emph{matching return} of $i$ along $\sigma$.
For a position $i$ of $\sigma$,
the \emph{abstract successor of $i$ along $\sigma$}, denoted
  $\SUCC(\abs,\sigma,i)$, is defined as follows:
  \begin{compactitem}
  \item If $i$ is a call,  then
    $\SUCC(\abs,\sigma,i)$ is the matching return of $i$ if such a matching return exists; otherwise
     $\SUCC(\abs,\sigma,i)=\NULL$ ($\NULL$ denotes the
    \emph{undefined} value).
  \item If $i$ is not a call, then $\SUCC(\abs,\sigma,i)=i+1$ if $i+1<|\sigma|$ and $i+1$ is
    not a return position, and $\SUCC(\abs,\sigma,i)=\NULL$, otherwise.
  \end{compactitem}
The \emph{caller of $i$ along $\sigma$}, denoted   $\SUCC(\caller,\sigma,i)$, is instead defined as follows:
\begin{compactitem}
  \item if there exists the greatest call position $j_c<i$ such that either   $\SUCC(\abs,\sigma,j_c)=\NULL$ or
  $\SUCC(\abs,\sigma,j_c)>i$, then  $\SUCC(\caller,\sigma,i)=j_c$; otherwise, $\SUCC(\caller,\sigma,i)=\NULL$.
  \end{compactitem}
We also consider  the \emph{global successor} $\SUCC(\Global,\sigma,i)$ of $i$ along $\sigma$ given by   $i+1$ if $i+1<|\sigma|$,
and undefined otherwise.
A \emph{maximal abstract path} (\MAP) of $\sigma$ is a \emph{maximal} (finite or infinite) increasing sequence
of natural numbers $\nu= i_0<i_1<\ldots$ such that
$i_j=\SUCC(\abs,\sigma,i_{j-1})$ for all $1\leq j<|\nu|$. Note that for every position $i$ of $\sigma$, there is exactly one \MAP\ of $\sigma$ visiting
position $i$.  For each $i\geq 0$, the \emph{caller path of $\sigma$ from position $i$} is the maximal
(finite) decreasing sequence of natural numbers $j_0>j_1\ldots >j_n$ such that $j_0= i$ and $j_{h+1}=\SUCC(\caller,\sigma,j_{h})$ for all $0\leq h <n$.
Note that all the positions of a \MAP\ have the same caller (if any). Intuitively, in the analysis of recursive programs, a maximal  abstract path
captures the local computation within a
procedure removing computation fragments corresponding to nested
calls, while the caller path represents the call-stack
content at a given position of the input.

For instance, consider the finite untimed word $\sigma$ of length $10$ depicted in Figure~\ref{fig-word}
where $\Scall =\{c\}$, $\Sret =\{r\}$, and $\Sint =\{\imath\}$. Note that $0$ is the unique unmatched call position of $\sigma$: hence, the \MAP\ visiting $0$ consists of just position $0$ and has no caller. The \MAP\ visiting position $1$ is the   sequence $1,6,7,9,10$ and the associated caller is position $0$.
The \MAP\ visiting position $2$ is the sequence $2,3,5$ and the associated caller is position $1$, and the \MAP\ visiting position $4$ consists of just position $4$ whose caller path is $4,3,1,0$.

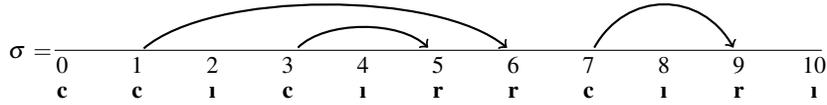
\begin{figure}[H]
\centering
\vspace{-0.2cm}
\begin{tikzpicture}[scale=1]

\node (Word) at (-0.2,0.2) {};
\coordinate [label=left:{\footnotesize  $\sigma$\,\,$=$}] (Word) at (0.0,0.17);
\path[thin,black] (-0.1,0.2) edge   (10.1,0.2);

\node (NodeZero) at (0.0,0.0) {};
\coordinate [label=center:{\footnotesize  $0$}] (NodeZero) at (0.0,0.0);
\coordinate [label=below:{\footnotesize   \textbf{c}}] (NodeZero) at (0.0,-0.15);

\node (NodeOne) at (1.0,0.0) {};
\coordinate [label=center:{\footnotesize  $1$}] (NodeOne) at (1.0,0.0);
\coordinate [label=below:{\footnotesize   \textbf{c}}] (NodeOne) at (1.0,-0.15);

\node (NodeTwo) at (2.0,0.0) {};
\coordinate [label=center:{\footnotesize  $2$}] (NodeTwo) at (2.0,0.0);
\coordinate [label=below:{\footnotesize   \textbf{\i}}] (NodeTwo) at (2.0,-0.15);

\node (SourceOne) at (0.95,0.11) {};
\node (SourceSix) at (6.05,0.12) {};
\node (SourceSeven) at (7.0,0.10) {};
\node (SourceNine) at (9.0,0.10) {};
\node (SourceThree) at (3.0,0.10) {};
\node (SourceFive) at (5.0,0.10) {};


\draw[->,thick,black] (SourceOne) .. controls (2.0,1.0) and  (5.0,1.0)  ..   (SourceSix);
\draw[->,thick,black] (SourceSeven) .. controls (7.5,1.0) and  (8.5,1.0)  ..   (SourceNine);

\node (NodeThree) at (3.0,0.0) {};
\coordinate [label=center:{\footnotesize  $3$}] (NodeThree) at (3.0,0.0);
\coordinate [label=below:{\footnotesize   \textbf{c}}] (NodeThree) at (3.0,-0.15);

\node (NodeFour) at (4.0,0.0) {};
\coordinate [label=center:{\footnotesize  $4$}] (NodeFour) at (4.0,0.0);
\coordinate [label=below:{\footnotesize   \textbf{\i}}] (NodeFour) at (4.0,-0.15);

\node (NodeFive) at (5.0,0.0) {};
\coordinate [label=center:{\footnotesize  $5$}] (NodeFive) at (5.0,0.0);
\coordinate [label=below:{\footnotesize   \textbf{r}}] (NodeFive) at (5.0,-0.15);

\draw[->,thick,black] (SourceThree) .. controls (3.5,0.6) and  (4.5,0.6)  .. node[above] {}  (SourceFive);

\node (NodeSix) at (6.0,0.0) {};
\coordinate [label=center:{\footnotesize  $6$}] (NodeSix) at (6.0,0.0);
\coordinate [label=below:{\footnotesize   \textbf{r}}] (NodeSix) at (6.0,-0.15);

\node (NodeSeven) at (7.0,0.0) {};
\coordinate [label=center:{\footnotesize  $7$}] (NodeSeven) at (7.0,0.0);
\coordinate [label=below:{\footnotesize   \textbf{c}}] (NodeSeven) at (7.0,-0.15);

\node (NodeEight) at (8.0,0.0) {};
\coordinate [label=center:{\footnotesize  $8$}] (NodeEight) at (8.0,0.0);
\coordinate [label=below:{\footnotesize   \textbf{\i}}] (NodeEight) at (8.0,-0.15);

\node (NodeNine) at (9.0,0.0) {};
\coordinate [label=center:{\footnotesize  $9$}] (NodeNine) at (9.0,0.0);
\coordinate [label=below:{\footnotesize   \textbf{r}}] (NodeNine) at (9.0,-0.15);

\node (NodeTen) at (7.0,0.0) {};
\coordinate [label=center:{\footnotesize  $10$}] (NodeTen) at (10.0,0.0);
\coordinate [label=below:{\footnotesize   \textbf{\i}}] (NodeTen) at (10.0,-0.15);

\end{tikzpicture}
\caption{\label{fig-word} An untimed word over a pushdown alphabet}
\vspace{-0.2cm}
\end{figure}

\section{Event-clock nested automata}

In this section, we recall the class of \emph{Event-Clock Nested Automata} (\ECNA)~\cite{BMP18}, a formalism that combines Event Clock Automata (\ECA)~\cite{AlurFH99} and Visibly Pushdown Automata (\VPA)~\cite{AlurMadhu04}
 by allowing a combined used of event clocks and visible operations on the stack. 

 Here, we adopt a propositional-based approach, where the pushdown alphabet is implicitly given. This is because in formal verification,
 one usually considers a finite set
  of atomic propositions which represent predicates over the states of the given system. Moreover, for verifying recursive programs,
one fixes three additional propositions, here denoted by $\call$, $\ret$, and $\intA$: $\call$ denotes the invocation of a procedure, $\ret$ denotes the return from a procedure,
and $\intA$ denotes internal actions of the current procedure. Thus, we fix a finite set $\Prop$ of atomic propositions containing the special propositions
$\call$, $\ret$, and $\intA$. The set $\Prop$ induces a pushdown alphabet $\Sigma_\Prop=\Scall\cup\Sret\cup\Sint$, where $\Scall =\{P\subseteq \Prop\mid P\cap \{\call,\ret,\intA\} = \{\call\} \}$, $\Sret =\{P\subseteq \Prop\mid P\cap \{\call,\ret,\intA\} = \{\ret\} \}$, and
$\Sint =\{P\subseteq \Prop\mid P\cap \{\call,\ret,\intA\} = \{\intA\} \}$.

 The set $C_{\Prop}$ of event clocks associated
with $\Prop$ is given by
$C_{\Prop}:= \bigcup_{p\in\Prop} \{x^{\Global}_p,y_p^{\Global},x_p^{\abs},y_p^{\abs},x_p^{\caller}\}$. Thus,  we associate with each
proposition $p\in \Prop$,  five event clocks:
the \emph{global recorder clock $x^{\Global}_p$} (resp., the \emph{global predictor clock $y^{\Global}_p$}) recording the time elapsed since the last occurrence of $p$ if any
(resp., the time required to the next occurrence of $p$ if any);  the \emph{abstract recorder clock $x_p^{\abs}$} (resp., the \emph{abstract predictor clock $y_p^{\abs}$}) recording the time elapsed since the last occurrence of $p$ if any (resp., the time required to the next occurrence of $p$) along the 
\MAP\ visiting the current position; and the \emph{caller (recorder) clock} $x_p^{\caller}$ recording the time elapsed since the last occurrence of $p$ if any along the caller path from the current position.
Let $w=(\sigma,\tau)$ be a timed word over $\Sigma_\Prop$ and $0\leq i< |w|$. We denote by  $\Pos(\abs,\sigma,i)$ the set of positions visited by the \MAP\ of $\sigma$ associated with position $i$, and
by  $\Pos(\caller,\sigma,i)$ the set of positions visited by the caller path of $\sigma$ from position $i$. For having 
a uniform notation, 
let $\Pos(\Global,\sigma,i)$ 
be the full set  of $w$-positions.
The values of the clocks at a 
position $i$ of the word $w$ can be deterministically determined as follows.

\begin{definition}[Determinisitic clock valuations] A \emph{clock valuation} over $C_{\Prop}$ is a mapping $\val: C_{\Prop} \mapsto \RealP\cup \{\NULL\}$, assigning to each event clock a value in
$\RealP\cup \{\NULL\}$ ($\NULL$ 
is the \emph{undefined} value).
For a  timed word $w=(\sigma,\tau)$ over
 $\Sigma$ and 
 $0\leq i <|w|$, the \emph{clock valuation $\val^{w}_i$ over $C_{\Prop}$}, specifying the values of the event clocks at position $i$ along $w$, is defined as follows for each $p\in\Prop$, where  $\dir\in \{\Global,\abs\}$ and $\dir' \in \{\Global,\abs, \caller\}$:
\[
\begin{array}{l}
 \val^{w}_i(x_p^{\dir'})  =  \left\{
\begin{array}{ll}
\tau_i - \tau_j
&    \text{ if there exists the unique } j<i:\, p\in \sigma_j,\, j\in\Pos(\dir',\sigma,i),  \text{ and}\,
\\
& \,\,\,\,\,\,\forall k: (j < k < i   \text{ and }\, k\in\Pos(\dir',\sigma,i))  \Rightarrow p\notin \sigma_k\\
\NULL
&    \text{ otherwise }
\end{array}
\right.
\vspace{0.2cm}\\
\val^{w}_i(y_p^{\dir})  =  \left\{
\begin{array}{ll}
\tau_j - \tau_i
&    \text{ if there exists the unique }   j>i:\, p\in \sigma_j,\, j\in\Pos(\dir,\sigma,i),  \text{ and}\,
\\
& \,\,\,\,\,\,\forall k: (i < k < j   \text{ and }\, k\in\Pos(\dir,\sigma,i))  \Rightarrow p\notin \sigma_k\\
\NULL
&    \text{ otherwise }
\end{array}
\right.
\end{array}
\]
\end{definition}

 It is worth noting that while the values of the global clocks are obtained by considering the full set of positions in $w$, the values of the abstract clocks (resp., caller clocks) are defined with respect to the \MAP\ visiting the current position (resp., with respect to the caller path from the current position).

 A \emph{clock constraint} over $C_{\Prop}$ is a conjunction of atomic formulas of the form
$z \in I$, where $z\in C_{\Prop}$, and
$I$ is either an interval in $\RealP$ with bounds in $\Nat\cup\{\infty\}$, or the singleton $\{\NULL\}$.
For a clock valuation $\val$ and a clock constraint $\theta$, $\val$ satisfies $\theta$, written
$\val\models \theta$, if for each conjunct $z\in I$ of $\theta$, $\val(z)\in I$. We denote by $\Phi(C_{\Prop})$ the set of clock constraints over $C_{\Prop}$.

\begin{definition} An   \ECNA\   over  $\Sigma_\Prop= \Scall\cup \Sint \cup \Sret$ is a tuple
$\Au=\tpl{\Sigma_\Prop, Q,Q_{0},  C_{\Prop},\Gamma\cup\{\bot\},\Delta,F}$, where $Q$ is a finite
set of (control) states, $Q_{0}\subseteq Q$ is a set of initial
states,  $\Gamma\cup\{\bot\}$ is a finite stack alphabet,  $\bot\notin\Gamma$ is the special \emph{stack bottom
  symbol}, $F\subseteq Q$ is a set of accepting states, and $\Delta=\Delta_c\cup \Delta_r\cup \Delta_i$ is a transition relation, where:
  \begin{compactitem}
    \item $\Delta_c\subseteq Q\times \Scall \times \Phi(C_{\Prop})  \times Q \times \Gamma$ is the set of \emph{push transitions},
    \item $\Delta_r\subseteq Q\times \Sret  \times \Phi(C_{\Prop})  \times (\Gamma\cup \{\bot\}) \times Q $ is the set of \emph{pop transitions},
       \item $\Delta_i\subseteq Q\times \Sint\times \Phi(C_{\Prop})  \times Q $ is the set of \emph{internal transitions}.
  \end{compactitem}
\end{definition}

 We  now describe how an \ECNA\ $\Au$ behaves over a  timed word $w$. Assume that on reading the $i$-th position of $w$, the current state of $\Au$ is $q$, and  $\val^{w}_i$ is the event-clock valuation associated with  $w$ and position
 $i$. If $\Au$ reads a call $c\in \Scall$,  it chooses a push transition of the
form $(q,c,\theta, q',\gamma)\in\Delta_c$ and pushes the symbol $\gamma\neq \bot$ onto the
stack. If $\Au$ reads a return $r\in\Sret$,  it chooses a pop transition of the
form $(q,r,\theta, \gamma,q')\in\Delta_r$ such that $\gamma$ is the symbol on the top of the stack, and
pops $\gamma$ from the stack (if $\gamma=\bot$, then $\gamma$ is read but not removed). Finally, on reading an internal action $a\in\Sint$, $\Au$ chooses an internal transition of the
form $(q,a,\theta, q')\in\Delta_i$, and, in this case, there is no operation on the stack. Moreover, in all the cases, the constraint $\theta$
of the chosen transition must be fulfilled  by the  valuation $ \val^{w}_i$ and the control changes from $q$ to $q'$.

Formally, a configuration of $\Au$ is a pair $(q,\beta)$, where $q\in Q$ and
$\beta\in\Gamma^*\cdot\{\bot\}$ is a stack content.
A run $\pi$ of $\Au$ over a timed word $w=(\sigma,\tau)$
is a sequence  of configurations
 $\pi=(q_0,\beta_0),(q_1,\beta_1),\ldots
$ of length $|w|+1$ ($\infty+1$ stands for $\infty$) such that $q_0\in Q_{0}$,
$\beta_0=\bot$  (initialization), 
and the following holds
for all $0\leq i< |w|$:

\begin{description}
\item [Push] If $\sigma_i \in \Scall$, then for some $(q_i,\sigma_i,\theta,q_{i+1},\gamma)\in\Delta_c$,
  $\beta_{i+1}=\gamma\cdot \beta_i$ and $ \val^{w}_i \models \theta$.
\item [Pop] If $\sigma_i \in \Sret$, then for some
  $(q_i,\sigma_i,\theta, \gamma,q_{i+1})\in\Delta_r$,   $ \val^{w}_i \models \theta$,  and \emph{either}
  $\gamma\neq\bot$ and $\beta_{i}=\gamma\cdot \beta_{i+1}$, \emph{or}
  $\gamma=\beta_{i}= \beta_{i+1}=\bot$.
\item [Internal] If $\sigma_i \in \Sint$, then for some
  $(q_i,\sigma_i,\theta,q_{i+1})\in\Delta_i$, $\beta_{i+1}=\beta_i$ and $ \val^{w}_i\models \theta$.
\end{description}

The run $\pi$ is \emph{accepting} if \emph{either} $\pi$ is finite and $q_{|w|}\in F$, \emph{or} $\pi$ is infinite and there are infinitely many positions $i\geq 0$ such that $q_i\in F$.
The \emph{timed language} $\TLang(\Au)$ (resp., \emph{$\omega$-timed language} $\TLangInf(\Au)$) of $\Au$ is the set of finite (resp., infinite) timed words $w$ over $\Sigma_\Prop$
 such that there is an accepting run of $\Au$ on $w$. When considered as an acceptor of infinite timed words, an \ECNA\ is  called B\"{u}chi \ECNA.
 In this case, for technical convenience, we also consider  \ECNA\ equipped with a \emph{generalized B\"{u}chi acceptance condition} $\mathcal{F}$ consisting of a family of sets of accepting states. In such a setting, an infinite run $\pi$ is accepting if for each B\"{u}chi component $F\in \mathcal{F}$, the run $\pi$ visits infinitely often states in $F$.

In the following, we also consider the class of \emph{Visibly Pushdown Timed Automata} (\VPTA)~\cite{BouajjaniER94,EmmiM06},  a combination of \VPA\ and standard Timed Automata~\cite{AlurD94}.
The clocks in a \VPTA\  can be reset when a transition is taken; hence, their values 
at a position of an input word depend  in general on the behaviour of the automaton and not only, as for event clocks, on the word. The syntax and semantics of \VPTA\ is shortly recalled in Appendix A of ~\cite{compGandalf18}.

\section{The Event-Clock Nested Temporal Logic}

A known decidable timed temporal logical framework related to the class of Event-Clock automata (\ECA) is the so called \emph{Event-Clock Temporal Logic} (\ECTL)~\cite{RaskinS99}, an extension
of standard \LTL\ with past obtained by means of two indexed modal operators $\PrevClock$  and $\NextClock$  which express real-time constraints.
On the other hand, for the class of \VPA, a related logical framework is the temporal logic 
\CARET\ \cite{AlurEM04}, a well-known context-free extension
of \LTL\ with past by means of non-regular versions of the \LTL\ temporal operators.
In this section, we introduce  an extension of
both \ECTL\ and \CARET, called \emph{Event-Clock Nested Temporal Logic} (\ECNTL) which allows to specify non-regular context-free real-time properties.

For the given set $\Prop$ of atomic propositions containing the special propositions $\call$, $\ret$, and $\intA$, the syntax of
\ECNTL\ formulas $\varphi$ is as follows:
\[
\varphi:= \true \DefORmini
p \DefORmini
\varphi \vee \varphi \DefORmini
\neg\,\varphi      \DefORmini
\Next^{\dir} \varphi     \DefORmini
\Prev^{\dir'} \varphi       \DefORmini
\varphi\,\Until^{\dir}\varphi  \DefORmini
\varphi\,\Since^{\dir'}\varphi \DefORmini
\NextClock^{\dir}_I \varphi \DefORmini
\PrevClock^{\dir'}_I \varphi
\]
where $p\in \Prop$, $I$ is an interval in $\RealP$ with bounds in $\Nat\cup\{\infty\}$, $\dir\in\{\Global,\abs\}$, and $\dir'\in \{\Global,\abs,\caller\}$. The operators
$\Next^{\Global}$,
$\Prev^{\Global}$, $\Until^{\Global}$, and $\Since^{\Global}$ are the standard `next',
`previous', `until', and `since' \LTL\ modalities, respectively,
$\Next^{\abs}$, $\Prev^{\abs}$, $\Until^{\abs}$, and
$\Since^{\abs}$  are their
non-regular abstract versions, and $\Prev^{\caller}$ and $\Since^{\caller}$  are the non-regular caller versions of the  `previous' and `since' \LTL\ modalities.
Intuitively, the abstract and caller
 modalities allow  to specify \LTL\ requirements on the abstract and caller paths of the given timed word over $\Sigma_\Prop$.
Real-time constraints are specified by the indexed operators $\NextClock_I^{\Global} $, $\PrevClock_I^{\Global} $, $\NextClock^{\abs}_I $,  $\PrevClock^{\abs}_I $, and $\PrevClock^{\caller}_I $.
The formula $\NextClock_I^{\Global} \varphi$  requires that the delay $t$ before the next position where $\varphi$ holds satisfies $t\in I$; symmetrically,
$\PrevClock_I^{\Global} \varphi$  constraints the previous position where $\varphi$ holds.  The abstract versions
  $\NextClock^{\abs}_I \varphi$ and $\PrevClock^{\abs}_I \varphi$ are similar, but the notions of
  next and previous position where $\varphi$ holds refer to the \MAP\ visiting the current position. Analogously,
  for the caller version $\PrevClock^{\caller}_I \varphi$ of $\PrevClock_I^{\Global} \varphi$, the notion  of
    previous position where $\varphi$ holds refers to the caller path visiting the current position.

 Full \CARET\ \cite{AlurEM04} corresponds to the fragment of \ECNTL\ obtained by disallowing the real-time operators, while the logic  \ECTL\  \cite{RaskinS99} is obtained
 from \ECNTL\   by disallowing the abstract and caller modalities.
   As pointed out in~\cite{RaskinS99}, the real-time operators $\PrevClock$ and $\NextClock$ generalize the semantics of event clock variables since they allows recursion, i.e., they can constraint
   arbitrary formulas and not only atomic propositions. Accordingly, the \emph{non-recursive fragment} of \ECNTL\ is obtained by replacing the clauses $\NextClock^{\dir}_I \varphi$ and
$\PrevClock^{\dir'}_I \varphi$ in the syntax with the clauses $\NextClock^{\dir}_I p$ and
$\PrevClock^{\dir'}_I p$, where $p\in \Prop$.
We use standard shortcuts in \ECNTL: the formula $\Eventually^{\Global} \psi$ stands for $\true\,\Until^{\Global}\, \psi$
(the  \LTL\ eventually operator), and $\Always^{\Global} \psi$ stands for $\neg \Eventually^{\Global} \neg\psi$ (the  \LTL\ always
operator). For an \ECNTL\ formula $\varphi$, $|\varphi|$ denotes the number of distinct subformulas of $\varphi$ and $\Const_\varphi$ the set of constants used as finite endpoints
in the intervals associates with the real-time modalities. The size of $\varphi$ is $|\varphi| + k$, where $k$ is the size of the binary encoding of the largest constant in
$\Const_\varphi$.

Given an \ECNTL\ formula $\varphi$, a timed word $w=(\sigma,\tau)$ over $\Sigma_\Prop$  and a position $0\leq i< |w|$, the satisfaction relation
$(w,i)\models\varphi$ is inductively defined as follows (we omit the clauses for the atomic propositions and Boolean connectives which are standard):
\[ \begin{array}{ll}
  (w,i)\models \Next^{\dir}\varphi &
              \Leftrightarrow\,  \text{ there is } j>i \text{ such that } j= \SUCC(\dir,\sigma,i) \text{ and } (w,j)\models \varphi
    \\
 (w,i)\models \Prev^{\dir'}\varphi &
              \Leftrightarrow\,  \text{ there is } j<i \text{ such that } (w,j)\models \varphi \text{ and \emph{either} } (dir'\neq \caller  \text{ and }
  \\ & \phantom{\Leftrightarrow}\,\,  i= \SUCC(\dir',\sigma,j)), \text{ or } (dir'= \caller  \text{ and } j= \SUCC(\caller,\sigma,i))
\\
  (w,i)\models \varphi_1  \Until^{\dir}\varphi_2 &
              \Leftrightarrow\,  \textrm{there is   }  j\geq  i \text{ such that }j\in \Pos(\dir,\sigma,i),\, (w,j)\models \varphi_2 \text{ and }
       \\ & \phantom{\Leftrightarrow}\,\,     (w,k)\models \varphi_1 \text{ for all } k\in [i,j-1]\cap \Pos(\dir,\sigma,i)  \\
    (w,i)\models \varphi_1  \Since^{\dir'}\varphi_2 &
                      \Leftrightarrow\,  \textrm{there is   }  j\leq  i \text{ such that }j\in \Pos(\dir',\sigma,i),\, (w,j)\models \varphi_2 \text{ and }
       \\ & \phantom{\Leftrightarrow}\,\,     (w,k)\models \varphi_1 \text{ for all } k\in [j+1,i]\cap \Pos(\dir',\sigma,i)  \\
\end{array}
\]
\[ \begin{array}{ll}
   (w,i)\models \NextClock_I^{\dir}\varphi &
              \Leftrightarrow\,  \textrm{there is   }  j>  i \text{ s.t. }j\in \Pos(\dir,\sigma,i),\, (w,j)\models \varphi,\,\tau_j-\tau_i\in I,
       \\ & \phantom{\Leftrightarrow}\,\,  \text{ and }     (w,k)\not\models \varphi  \text{ for all } k\in [i+1,j-1]\cap \Pos(\dir,\sigma,i)  \\
   (w,i)\models \PrevClock_I^{\dir'}\varphi &
                 \Leftrightarrow\,  \textrm{there is   }  j<  i \text{ s.t. }j\in \Pos(\dir',\sigma,i),\, (w,j)\models \varphi,\,\tau_i-\tau_j\in I,
       \\ & \phantom{\Leftrightarrow}\,\,  \text{ and }     (w,k)\not\models \varphi  \text{ for all } k\in [j+1,i-1]\cap \Pos(\dir',\sigma,i)
\end{array}
\]
A timed word $w$ satisfies a formula $\varphi$ (we also say that $w$ is a model of $\varphi$) if $(w,0)\models \varphi$. The timed language $\TLang(\varphi)$ (resp. $\omega$-timed language $\TLangInf(\varphi)$) of $\varphi$ is the set of finite (resp., infinite) timed words over $\Sigma_\Prop$ satisfying $\varphi$. We consider the following decision problems:
\begin{compactitem}
  \item \emph{Satisfiability:} has a given \ECNTL\ formula a finite (resp., infinite) model?
  \item \emph{Visibly model-checking:} given a \VPTA\ $\Au$ over $\Sigma_\Prop$ 
  and an \ECNTL\ formula $\varphi$ over $\Prop$, does $\TLang(\Au)\subseteq \TLang(\varphi)$ (resp., $\TLangInf(\Au)\subseteq \TLangInf(\varphi)$) hold?
\end{compactitem}\vspace{0.1cm}


The logic \ECNTL\  allows to express in a natural way real-time \LTL-like properties over the non-regular patterns capturing the local computations of procedures or
the stack contents at given positions. Here, we consider three relevant examples.
\begin{compactitem}
  \item \emph{Real-time total correctness:} a bounded-time total correctness requirement for a procedure $A$ specifies that if the pre-condition $p$ holds when the procedure $A$
  is invoked, then the procedure must return within $k$ time units and $q$ must hold upon return. Such a requirement can be expressed by the following non-recursive formula, where proposition $p_A$
  characterizes calls to procedure $A$: $ \Always^{\Global}\bigl((\call\wedge p\wedge p_A) \rightarrow (\Next^{\abs}q \wedge \NextClock_{[0,k]}^{\abs} \ret)\bigr)$

  \item \emph{Local bounded-time response properties:} the requirement that in the local computation (abstract path) of a procedure $A$, every request $p$ is followed
  by a response $q$ within $k$ time units can be expressed by the following non-recursive formula, where $c_A$ denotes that the control is inside procedure $A$:
    $\Always^{\Global}\bigl((  p\wedge c_A) \rightarrow   \NextClock_{[0,k]}^{\abs} q \bigr)$
  \item \emph{Real-time properties  over the stack content:}
the real-time security requirement
that a procedure $A$ is invoked only if procedure $B$ belongs to the call
stack and within $k$ time units since the activation of $B$ can be expressed as follows (the calls to procedure $A$ and $B$ are marked by proposition $p_A$ and $p_B$, respectively):
%
$  \Always^{\Global}\bigl(( \call  \wedge p_A) \rightarrow   \PrevClock_{[0,k]}^{\caller} \,p_B \bigr)$
   %
\end{compactitem}\vspace{0.2cm}

\noindent{\textbf{Expressiveness results.}} We now compare the expressive power of the formalisms \ECNTL, \ECNA, and \VPTA\ with  respect to the associated classes of ($\omega$-)timed languages.
It is known that \ECA\ and the logic  \ECTL\  are expressively incomparable~\cite{RaskinS99}. This result  trivially generalizes to  \ECNA\ and \ECNTL\
(note that over timed words consisting only of internal actions, \ECNA\ correspond to \ECA, and the logic \ECNTL\ corresponds to \ECTL). In~\cite{BMP18}, it is shown that
\ECNA\ are strictly less expressive than \VPTA. In Section~\ref{sec:DecisionProceduresECNTL}, we show that \ECNTL\ is subsumed by  \VPTA\ (in particular, every \ECNTL\ formula can be translated into an equivalent
\VPTA). The inclusion is strict since the logic \ECNTL\ is closed under complementation,  while   \VPTA\ are not~\cite{EmmiM06}. Hence, we obtain the following result.

 \begin{theorem} Over finite (resp., infinite) timed words, \ECNTL\ and  \ECNA\ are expressively incomparable, and   \ECNTL\ is strictly less expressive than \VPTA.
 \end{theorem}

We additionally investigate the expressiveness of the novel timed temporal modalities $\PrevClock^{\abs}_I$, $\NextClock^{\abs}_I$, and $\PrevClock^{\caller}_I$. It turns out  that these
modalities add expressive power. 

 \begin{theorem}\label{theo:expressofNovelModalities} Let $\Fragm$ be the fragment of \ECNTL\ obtained by disallowing the
   modalities $\PrevClock^{\abs}_I$, $\PrevClock^{\caller}_I$, and $\NextClock^{\abs}_I$. Then, $\Fragm$ is strictly less expressive than \ECNTL.
 \end{theorem}
 \begin{proof} We focus on the case of finite timed words (the case of infinite timed words is similar). Let $\Prop =\{\call,\ret\}$ and $\TLang$ be the timed language
 consisting of the finite  timed words  of the form $(\sigma,\tau)$ such that $\sigma$ is a well-matched word of the form $ \{\call\}^{n} \cdot   \{\ret\}^{n}$ for some $n>0$, and there is a call
  position $i_c$ of $\sigma$ such that $\tau_{i_r}-\tau_{i_c} = 1$, where $i_r$ is the matching-return of $i_c$ in $\sigma$.
  $\TLang$ can be easily expressed in \ECNTL. On the other hand, one can show that $\TLang$ is not definable in $\Fragm$ (the detailed proof can be found in Appendix B of \cite{compGandalf18}) .
 \end{proof}

\subsection{Decision procedures for the logic \ECNTL}\label{sec:DecisionProceduresECNTL}

In this section, we provide an automata-theoretic approach for solving satisfiability and  visibly model-checking   for the logic \ECNTL\ which generalizes both the
automatic-theoretic approach of   \CARET\ \cite{AlurEM04} and the one for   \ECTL\ \cite{RaskinS99}. We focus on infinite timed words (the approach for finite timed words is similar). Given an \ECNTL\ formula $\varphi$ over $\Prop$, we construct in 
exponential time a generalized B\"{u}chi
$\ECNA$ $\Au_\varphi$ over an extension of the pushdown alphabet $\Sigma_\Prop$ accepting suitable encodings of the infinite models of $\varphi$.

Fix an \ECNTL\ formula $\varphi$ over $\Prop$. 
For each infinite timed word $w=(\sigma,\tau)$ over $\Sigma_\Prop$ we associate to $w$ an infinite timed word $\pi= (\sigma_e,\tau)$ over an extension of $\Sigma_\Prop$,
called \emph{fair Hintikka sequence}, where $\sigma_e = A_0 A_1\ldots$, and for all $i\geq 0$, $A_i$ is an \emph{atom} which, intuitively, describes a maximal set of subformulas of $\varphi$ which hold at position  $i$
along $w$. The notion of \emph{atom} syntactically captures the semantics of the Boolean connectives and
the local fixpoint characterization of the  variants of until (resp.,  since)
modalities 
 in terms of the corresponding variants of the next (resp., previous)  modalities.
Additional requirements on the timed word $\pi$, which can be easily checked by the transition function of an $\ECNA$, capture the semantics of the various next
and  previous modalities, and the semantics of the real-time operators. Finally, the global \emph{fairness} requirement, which can be easily checked  by a standard generalized
B\"{u}chi acceptance condition, captures the liveness requirements $\psi_2$ in until subformulas of the form $\psi_1\Until^{\Global} \psi_2$  (resp., $\psi_1\Until^{\abs} \psi_2$) of $\varphi$.
In particular, when an abstract until formula $\psi_1\Until^{\abs} \psi_2$ is asserted at   a position $i$ along an infinite timed word $w$ over $\Sigma_\Prop$ and the $\MAP$ $\nu$ visiting
position $i$ is infinite,
we have to ensure that the liveness requirement $\psi_2$ holds at some position $j\geq i$ of the $\MAP$ $\nu$. To this end,  we use a special proposition
$p_{\infty}$ which \emph{does not} hold  at a position $i$ of $w$ iff   position $i$ has a caller whose matching return is defined.  
We now proceed with the technical details.
The closure $\Cl(\varphi)$ of $\varphi$ is the smallest set containing:
\begin{compactitem}
  \item $\true\in \Cl(\varphi)$, each proposition $p\in \Prop\cup \{p_{\infty}\}$, and  formulas $\Next^{\abs}\true$ and $\Prev^{\abs}\true$;
  \item all the subformulas of $\varphi$;
  \item the formulas $\Next^{\dir}(\psi_1\Until^{\dir}\psi_2)$ (resp., $\Prev^{\dir}(\psi_1\Since^{\dir}\psi_2)$) for all the subformulas
  $\psi_1\Until^{\dir}\psi_2$ (resp., $\psi_1\Since^{\dir}\psi_2$) of $\varphi$, where $\dir\in \{\Global,\abs\}$ (resp., $\dir\in \{\Global,\abs,\caller\}$).
  \item all the negations of the above formulas (we identify $\neg\neg\psi$ with $\psi$).
\end{compactitem}\vspace{0.1cm}

Note that $\varphi\in\Cl(\varphi)$ and $|\Cl(\varphi)|=O(|\varphi|)$. In the following, elements of $\Cl(\varphi)$ are seen as atomic propositions, and we consider the pushdown alphabet $\Sigma_{\Cl(\varphi)}$ induced by
$\Cl(\varphi)$. In particular, for a timed word $\pi$ over $\Sigma_{\Cl(\varphi)}$, we consider the clock valuation $\val^{\pi}_i$  specifying the values of the event clocks $x_\psi$, $y_\psi$,
 $x^{\abs}_\psi$, $y^{\abs}_\psi$, and $x^{\caller}_\psi$  at
position $i$ along $\pi$, where $\psi\in \Cl(\varphi)$.

\noindent An \emph{atom $A$} of $\varphi$ is a subset of $\Cl(\varphi)$ satisfying the following:
\begin{compactitem}
  \item $A$ is a maximal subset of $\Cl(\varphi)$ which is propositionally consistent, i.e.:
  \begin{compactitem}
    \item $\true\in A$ and for each $\psi\in \Cl(\varphi)$, $\psi\in A$ iff $\neg\psi\notin A$;
    \item for each $\psi_1\vee\psi_2\in \Cl(\varphi)$, $\psi_1\vee\psi_2\in A$ iff $\{\psi_1,\psi_2\}\cap A \neq \emptyset$;
    \item $A$ contains exactly one atomic proposition in $\{\call,\ret,\intA\}$.
  \end{compactitem}
    \item for all $\dir\in \{\Global,\abs\}$ and $\psi_1\Until^{\dir}\psi_2\in \Cl(\varphi)$, either $\psi_2\in A$ or $\{\psi_1,\Next^{\dir}(\psi_1\Until^{\dir}\psi_2)\}\subseteq A$.
        \item for all $\dir\in \{\Global,\abs,\caller\}$ and $\psi_1\Since^{\dir}\psi_2\in \Cl(\varphi)$, either $\psi_2\in A$ or $\{\psi_1,\Prev^{\dir}(\psi_1\Since^{\dir}\psi_2)\}\subseteq A$.
  \item if $\Next^{\abs}\true\notin A$, then for all $\Next^{\abs}\psi\in \Cl(\varphi)$, $\Next^{\abs}\psi\notin A$.
   \item if $\Prev^{\abs}\true\notin A$, then for all $\Prev^{\abs}\psi\in \Cl(\varphi)$, $\Prev^{\abs}\psi\notin A$.
\end{compactitem}\vspace{0.1cm}

We now introduce the notion of Hintikka sequence $\pi$ which corresponds to an infinite timed word over $\Sigma_{\Cl(\varphi)}$ satisfying additional constraints. These constraints  capture the semantics of the variants of next,   previous, and real-time  modalities, and (partially) the intended meaning of proposition $p_{\infty}$  along the associated timed word over $\Sigma_\Prop$ (the projection of $\pi$ over $\Sigma_\Prop\times \RealP$). For an atom $A$, let $\Caller(A)$ be the set of caller formulas $\Prev^{\caller}\psi$ in $A$.
For atoms $A$ and $A'$, we define a predicate $\NextPrev(A,A')$ which holds if the global next (resp., global previous) requirements in $A$ (resp., $A'$) are the ones that hold in $A'$ (resp., $A$), i.e.:
(i) for all $\Next^{\Global}\psi\in \Cl(\varphi)$, $\Next^{\Global}\psi\in A$ iff  $\psi\in A'$, and (ii)
 for all $\Prev^{\Global}\psi\in \Cl(\varphi)$, $\Prev^{\Global}\psi\in A'$ iff  $\psi\in A$. Similarly, the predicate $\AbsNextPrev(A,A')$ holds if:
 (i) for all $\Next^{\abs}\psi\in \Cl(\varphi)$, $\Next^{\abs}\psi\in A$ iff  $\psi\in A'$, and (ii)
 for all $\Prev^{\abs}\psi\in \Cl(\varphi)$, $\Prev^{\abs}\psi\in A'$ iff  $\psi\in A$, and additionally (iii) $\Caller(A)=\Caller(A')$.
 Note that for   $\AbsNextPrev(A,A')$ to hold we also require that the caller requirements in $A$ and $A'$ coincide
 consistently with the fact that the positions of  a \MAP\ have the same caller (if any).

\begin{definition} \label{Def:HintikkaSequence}  An infinite timed word $\pi= (\sigma,\tau)$ over $\Sigma_{\Cl(\varphi)}$, where $\sigma =A_0 A_1\ldots$, is an \emph{Hintikka sequence of $\varphi$},  if for all $i\geq 0$, $A_i$ is a $\varphi$-atom and the following holds:
 \begin{compactenum}
    \item  \emph{Initial consistency:} for all $\dir\in\{\Global,\abs,\caller\}$ and $\Prev^{\dir}\psi\in\Cl(\varphi)$, $\neg \Prev^{\dir}\psi\in A_0$.
  \item  \emph{Global next and previous requirements:} $\NextPrev(A_i,A_{i+1})$.
  \item  \emph{Abstract and caller requirements:} we distinguish three cases.
  \begin{compactitem}
    \item $\call\notin A_i$ and $\ret\notin A_{i+1}$: $\AbsNextPrev(A_i,A_{i+1})$, $(p_{\infty}\in A_i$ iff $p_{\infty}\in A_{i+1})$;
    \item $\call\notin A_i$ and $\ret\in A_{i+1}$: $\Next^{\abs}\true\notin A_i$, and ($\Prev^{\abs}\true\in A_{i+1}$  iff the matching call of the return position $i+1$ is defined).
Moreover, if  $\Prev^{\abs}\true\notin A_{i+1}$, then $p_{\infty}\in A_i\cap A_{i+1}$ and $\Caller(A_{i+1})=\emptyset$.
    \item  $\call\in A_i$: if  $\SUCC(\abs,\sigma,i)=\NULL$ then $\Next^{\abs}\true\notin A_i$ and  $p_{\infty}\in A_i$; otherwise
    $\AbsNextPrev(A_i,A_{j})$ and $(p_{\infty}\in A_i$ iff $p_{\infty}\in A_{j})$, where $j=\SUCC(\abs,\sigma,i)$. Moreover, if
   $\ret\notin A_{i+1}$, then  $\Caller(A_{i+1})=\{\Prev^{\caller}\psi\in\Cl(\varphi)\mid\psi\in A_i\}$ and
   ($\Next^{\abs}\true\in A_i$ iff $p_{\infty}\notin A_{i+1}$).
  \end{compactitem}
  \item  \emph{Real-time requirements:}
  \begin{compactitem}
    \item for all $\dir\in\{\Global,\abs,\caller\}$ and $\PrevClock^{\dir}_I \psi \in\Cl(\varphi)$, $\PrevClock^{\dir}_I \psi \in A_i$ iff $\val_i^{\pi}(x^{\dir}_{\psi})\in I$;
        \item for all $\dir\in\{\Global,\abs\}$ and $\NextClock^{\dir}_I \psi \in\Cl(\varphi)$, $\NextClock^{\dir}_I \psi \in A_i$ iff $\val_i^{\pi}(y^{\dir}_{\psi})\in I$.
  \end{compactitem}
\end{compactenum}
\end{definition}

\noindent In order to capture the liveness requirements of the global and abstract  until subformulas of $\varphi$, and fully capture the intended meaning of proposition $p_{\infty}$, we consider the following additional global
fairness constraint. An Hintikka sequence $\pi=(A_0,t_0)(A_1,t_1)$ of $\varphi$ is \emph{fair} if
\begin{inparaenum}[(i)]
    \item for infinitely many $i\geq 0$, $p_{\infty}\in A_i$;
        \item for all $\psi_1\Until^{\Global} \psi_2\in\Cl(\varphi)$, there are infinitely many $i\geq 0$  s.t.  $\{\psi_2,\neg(\psi_1\Until^{\Global} \psi_2)\}\cap A_i\neq \emptyset$; and
  \item  for all $\psi_1\Until^{\abs} \psi_2\in\Cl(\varphi)$, there are infinitely many $i\geq 0$ such that $p_{\infty}\in A_i$ and  $\{\psi_2,\neg(\psi_1\Until^{\abs} \psi_2)\}\cap A_i\neq \emptyset$.
\end{inparaenum}

The Hintikka sequence $\pi$ is \emph{initialized} if $\varphi\in A_0$. Note that according to the intended meaning of proposition $p_{\infty}$, for each infinite timed word $w=(\sigma,\tau)$ over $\Sigma_\Prop$, $p_{\infty}$ holds at infinitely many positions. Moreover,
there is at a most one \emph{infinite} \MAP\ $\nu$ of $\sigma$, and for such a \MAP\ $\nu$ and each position $i$ greater than the starting position of $\nu$, either $i$ belongs to $\nu$ and $p_{\infty}$ holds, or $p_{\infty}$ does not hold. Hence, the fairness requirement for an abstract until subformula $\psi_1\Until^{\abs} \psi_2$ of $\varphi$ ensures that whenever $\psi_1\Until^{\abs} \psi_2$ is asserted at some position $i$ of $\nu$, then $\psi_2$ eventually holds at some position $j\geq i$ along $\nu$. Thus, we obtain the following characterization of the infinite models of $\varphi$, where
$\Proj_\varphi$ is the mapping associating to each  \emph{fair} Hintikka sequence $\pi=(A_0,t_0)(A_1,t_1)\ldots$ of $\varphi$, the infinite timed word over $\Sigma_\Prop$ given by
$\Proj(\pi)=(A_0\cap\Prop,t_0)(A_1\cap\Prop,t_1)\ldots$.

\begin{proposition}\label{prop:TimedHintikkaSequence}
\emph{Let $\pi=(A_0,t_0)(A_1,t_1)\ldots$ be a fair Hintikka sequence of $\varphi$. Then,  for all $i\geq 0$ and $\psi\in \Cl(\varphi)\setminus \{p_{\infty},\neg p_{\infty}\}$,
$\psi\in A_i$ \emph{iff} $(\Proj_\varphi(\pi),i)\models \psi$. Moreover, the mapping $\Proj_\varphi$ is a bijection between the set of fair  Hintikka sequences of $\varphi$ and the set of infinite timed words over $\Sigma_\Prop$. In particular, an infinite timed word over $\Sigma_\Prop$ is a model of $\varphi$ iff the associated fair Hintikka sequence is initialized.}
 \end{proposition}

The detailed proof of Proposition~\ref{prop:TimedHintikkaSequence} can be found in Appendix C of \cite{compGandalf18}. 

The notion of initialized fair Hintikka sequence can be easily captured by a generalized B\"{u}chi \ECNA.

\begin{theorem}\label{theo:FromECNTLtoECNA} Given an \ECNTL\ formula $\varphi$, one can construct in singly exponential time a generalized B\"{u}chi \ECNA\ $\Au_\varphi$
 having $  2^{O(|\varphi|)}$ states, $2^{O(|\varphi|)}$ stack symbols, a set of constants $\Const_\varphi$, and $O(|\varphi|)$ clocks. If $\varphi$ is non-recursive, then
 $\Au_\varphi$   accepts the infinite models of $\varphi$; otherwise, $\Au_\varphi$  accepts the set of initialized fair Hintikka sequences of $\varphi$.
 \end{theorem}
\begin{proof}
We first build a  generalized B\"{u}chi  \ECNA\ $\Au_\varphi$ over $\Sigma_{\Cl(\varphi)}$ accepting the set of initialized fair Hintikka sequences of $\varphi$.
 The set of $\Au_\varphi$ states is the set of atoms of $\varphi$, and a state $A_0$ is initial if $\varphi\in A_0$ and $A_0$ satisfies Property~1 (initial consistency) in Definition~\ref{Def:HintikkaSequence}.
In the transition function, we require that the input symbol coincides with the source state in such a way that in a run, the sequence of control states corresponds to the untimed part of the input.
By the transition function, the automaton checks that the input word is an Hintikka sequence. In particular, for the abstract next and abstract previous requirements (Property~3 in Definition~\ref{Def:HintikkaSequence}),
whenever the input symbol $A$ is a call, the automaton pushes on the stack the atom $A$. In such a way, on reading the matching return $A_r$ (if any) of the call $A$,  the automaton pops $A$ from the stack and can locally check that $\AbsNextPrev(A,A_r)$ holds. In order to ensure the real-time requirements (Property~4 in Definition~\ref{Def:HintikkaSequence}), $\Au_\varphi$ simply uses the  recorder clocks and   predictor clocks: a transition having as source state an atom  $A$ has a clock constraint whose set of atomic constraints has the form \vspace{-0.1cm}
\[
\displaystyle{\bigcup_{\PrevClock_I^{\dir}\psi\in A}\{x_{\psi}^{\dir}\in I\}\cup\bigcup_{\neg\PrevClock_I^{\dir}\psi\in A}\{x_{\psi}^{\dir}\in  \widehat{I}\}\cup
\bigcup_{\NextClock_I^{\dir}\psi\in A}\{y_{\psi}^{\dir}\in I\}\cup\bigcup_{\neg\NextClock_I^{\dir}\psi\in A}\{y_{\psi}^{\dir}\in \widehat{I}\}}\vspace{-0.2cm}
\]
where $\widehat{I}$ is either $\{\NULL\}$ or a \emph{maximal} interval over $\RealP$ disjunct from $I$.
Finally, the generalized B\"{u}chi acceptance condition is exploited for checking that the input initialized Hintikka sequence is fair. 
The detailed construction of
$\Au_\varphi$ can be found in Appendix D of \cite{compGandalf18}. Note that $\Au_\varphi$ has $2^{O(|\varphi|)}$ states and stack symbols, a set of constants $\Const_\varphi$, and $O(|\varphi|)$ event clocks.
If $\varphi$ is non-recursive, then the effective clocks are only associated with propositions in $\Prop$. Thus, by projecting the input symbols of the transition function of $\Au_\varphi$ over $\Prop$, by Proposition~\ref{prop:TimedHintikkaSequence},
we obtain a generalized B\"{u}chi \ECNA\ accepting the infinite models of $\varphi$.
\end{proof}

We can state now the main result of the section.

\begin{theorem}\label{theorem:complexityResultsECNTL} Given an \ECNTL\ formula $\varphi$  over $\Sigma_\Prop$, one can construct in singly exponential time a \VPTA,
with $2^{O(|\varphi|^{3})}$ states and stack symbols, $O(|\varphi|)$ clocks, and a set of constants $\Const_\varphi$, which accepts $\TLang(\varphi)$ (resp., $\TLangInf(\varphi)$). Moreover,
   satisfiability and visibly  model-checking for
 \ECNTL\ over finite (resp., infinite) timed words are \EXPTIME-complete.
 \end{theorem}
 \begin{proof} We focus on the case of infinite timed words. Fix an \ECNTL\ formula $\varphi$  over $\Sigma_\Prop$. By Theorem~\ref{theo:FromECNTLtoECNA},
 one can construct a generalized B\"{u}chi \ECNA\ $\Au_\varphi$ over $\Sigma_{\Cl(\varphi)}$
 having $  2^{O(|\varphi|)}$ states and stack symbols, a set of constants $\Const_\varphi$, and  accepting the set of initialized fair Hintikka sequences of $\varphi$.
 By~\cite{BMP18}, one can construct a generalized  B\"{u}chi \VPTA\ $\Au'_\varphi$ over $\Sigma_{\Cl(\varphi)}$ accepting $\TLangInf(\Au_\varphi)$, having
 $2^{O(|\varphi|^{2}\cdot k)}$ states and stack symbols, $O(k)$ clocks, and a set of constants $\Const_\varphi$,  where $k$ is the number of atomic constraints used by $\Au_\varphi$.
 Note that $k=O(|\varphi|)$. Thus, by projecting the input symbols of the transition function of $\Au'_\varphi$ over $\Prop$, we obtain a (generalized B\"{u}chi) \VPTA\ satisfying the first part of
 Theorem~\ref{theorem:complexityResultsECNTL}.

 For the upper bounds of the second part of  Theorem~\ref{theorem:complexityResultsECNTL},  observe that by~\cite{BouajjaniER94,AbdullaAS12} emptiness of generalized B\"{u}chi \VPTA\  is  solvable in time $O(n^{4} \cdot 2^{O(m\cdot \log K m)})$, where
$n$ is the number of states, $m$ is the number of clocks, and $K$ is the largest constant used in the clock constraints of the automaton (hence, the time complexity is polynomial in the number of states).
Now, given a B\"{u}chi \VPTA\ $\Au$ over $\Sigma_\Prop$ 
and an \ECNTL\ formula $\varphi$ over $\Sigma_\Prop$,   model-checking
$\Au$ against $\varphi$ reduces to check emptiness of 
$\TLangInf(\Au)\cap  \TLangInf(\Au'_{\neg\varphi})$, where $\Au'_{\neg\varphi}$ is the generalized B\"{u}chi \VPTA\
associated with $\neg\varphi$. Thus, since B\"{u}chi \VPTA\ are polynomial-time closed under intersection, membership in \EXPTIME\ for   satisfiability and visibly  model-checking of
 \ECNTL\ follow. The matching lower bounds 
 follow  from \EXPTIME-completeness of satisfiability and visibly model-checking for the logic \CARET\ \cite{AlurEM04} which is subsumed by \ECNTL. 
 \end{proof} 
\section{Nested Metric Temporal Logic (\NMTL)}

Metric temporal logic (\MTL) \cite{Koymans90} is a well-known  timed linear-time temporal logic which extends \LTL\ with time
constraints on until modalities. In this section, we introduce an extension of \MTL\ with past, we call \emph{nested} \MTL\ (\NMTL, for short), by means of timed versions of the
\CARET\ modalities.

For the given set $\Prop$ of atomic propositions containing the special propositions $\call$, $\ret$, and $\intA$, the syntax of
nested \NMTL\ formulas $\varphi$ is as follows:\vspace{-0.1cm}
\[
\varphi:= \top \DefORmini
p \DefORmini
\varphi \vee \varphi \DefORmini
\neg\,\varphi      \DefORmini
\varphi\,\StrictUntil^{\dir}_I\varphi  \DefORmini
\varphi\,\StrictSince^{\dir'}_I\varphi \vspace{-0.05cm}
\]
where $p\in \Prop$, $I$ is an interval in $\RealP$ with endpoints in $\Nat\cup\{\infty\}$, $\dir\in \{\Global,\abs\}$ and $\dir'\in\{\Global,\abs,\caller\}$. The operators
$\StrictUntil^{\Global}_I$  and $\StrictSince^{\Global}_I$ are the standard \emph{timed  until and timed since} \MTL\ modalities, respectively,
 $\StrictUntil^{\abs}_I$  and
$\StrictSince^{\abs}_I$  are their
non-regular abstract versions, and  $\StrictSince_I^{\caller}$  is the non-regular caller version  of $\StrictSince^{\Global}_I$.
\MTL\ with past  
is the fragment of  \NMTL\ obtained by disallowing the timed  abstract and caller modalities, while
standard \MTL\ or future \MTL\ is the fragment of \MTL\ with past where the global  timed since   modalities are disallowed.
For an \NMTL\ formula $\varphi$, a  timed word $w=(\sigma,\tau)$ over $\Sigma_\Prop$  and 
$0\leq  i< |w|$, the satisfaction relation
$(w,i)\models\varphi$ is  defined as follows
(we omit the clauses for  propositions and Boolean connectives):
\vspace{-0.1cm}
\[ \begin{array}{ll}
   (w,i)\models \varphi_1 \StrictUntil_I^{\dir}\varphi_2 &
              \Leftrightarrow\,  \textrm{there is   }  j>  i \text{ s.t. }j\in \Pos(\dir,\sigma,i),\, (w,j)\models \varphi_2,\,\tau_j-\tau_i\in I,
       \\ & \phantom{\Leftrightarrow}\,\,  \text{ and }     (w,k)\models \varphi_1  \text{ for all } k\in [i+1,j-1]\cap \Pos(\dir,\sigma,i)  \\
   (w,i)\models \varphi_1 \StrictSince_I^{\dir'}\varphi_2 &
                 \Leftrightarrow\,  \textrm{there is   }  j<  i \text{ s.t. }j\in \Pos(\dir',\sigma,i),\, (w,j)\models \varphi_2,\,\tau_i-\tau_j\in I,
       \\ & \phantom{\Leftrightarrow}\,\,  \text{ and }     (w,k)\models \varphi_1  \text{ for all } k\in [j+1,i-1]\cap \Pos(\dir',\sigma,i)
\end{array}\vspace{-0.1cm}
\]
\noindent In the following, we use some derived operators in \NMTL:
\begin{compactitem}
  \item For $\dir\in \{\Global,\abs\}$,  $\StrictEventually^{\dir}_I \varphi := \top \, \StrictUntil^{\dir}_I\varphi $ 
  and   $\StrictAlways^{\dir}_I \varphi:= \neg\StrictEventually^{\dir}_I\neg \varphi $
  \item for $\dir\in\{\Global,\abs,\caller\}$,  $\StrictPastEventually^{\dir}_I \varphi:= \top\, \StrictSince^{\dir}_I\varphi $ 
  and $\StrictPastAlways^{\dir}_I \varphi:= \neg\StrictPastEventually^{\dir}_I\neg \varphi $. 
\end{compactitem}\vspace{0.1cm}

Let  $\INTS$ be the set of \emph{nonsingular} intervals $J$ in $\RealP$ with endpoints in $\Nat\cup\{\infty\}$ such that either $J$ is unbounded, or $J$
  is left-closed with left endpoint $0$. Such intervals $J$ can be replaced by expressions of the form $\sim c$ for some $c\in\Nat$
  and $\sim\in\{<,\leq,>,\geq\}$.  We focus on the following two fragments of  \NMTL:
\begin{compactitem}
  \item \NMITLS: obtained by allowing only intervals in $\INTS$.
  \item \emph{Future} \NMTL: obtained by disallowing the variants of timed since modalities.
\end{compactitem}\vspace{0.1cm}

It is known that for the considered pointwise semantics, \MITLS\ \cite{AlurFH96} (the fragment of \MTL\ allowing only intervals in $\INTS$) and \ECTL\ are equally expressive~\cite{RaskinS99}. Here, we easily
generalize such a result to the nested extensions
of \MITLS\ and \ECTL.

\newcounter{lemma-globalEquivNMTL-ECNTL}
\setcounter{lemma-globalEquivNMTL-ECNTL}{\value{lemma}}

\begin{lemma}\label{lemma:globalEquivNMTL-ECNTL} There exist effective linear-time translations   from \ECNTL\ into \NMITLS, and vice versa.
\end{lemma}

A proof of Lemma~\ref{lemma:globalEquivNMTL-ECNTL} can be found in Appendix  E of \cite{compGandalf18}.  By
Lemma~\ref{lemma:globalEquivNMTL-ECNTL}
and Theorem~\ref{theorem:complexityResultsECNTL}, we obtain the following result.

 \begin{theorem} \ECNTL\ and \NMITLS\ are expressively equivalent. Moreover, satisfiability and visibly  model-checking for
 \NMITLS\ over finite (resp., infinite) timed words are \EXPTIME-complete.
 \end{theorem}



In the considered pointwise semantics setting, it is well-known that satisfiability of  \MTL\ with past is undecidable~\cite{AlurH93,OuaknineW06}. Undecidability already holds for future \MTL\ interpreted
over infinite timed words~\cite{OuaknineW06}. However, over finite timed words, satisfiability of future \MTL\ is instead decidable~\cite{OuaknineW07}. Here, we show that over finite timed words, the addition of the future abstract timed modalities to future
\MTL\   makes the satisfiability problem undecidable.

\begin{theorem}\label{theorem:undecidability} Satisfiability of future \NMTL\ interpreted over finite timed words is undecidable.
\end{theorem}

We prove Theorem~\ref{theorem:undecidability} by a reduction from the halting problem
for Minsky $2$-counter machines~\cite{Minsky67}. Fix such a machine $M$
which is a tuple $M = \tpl{\Lab,\Inst,\ell_\init,\ell_\halt}$, where $\Lab$ is a finite set of labels (or program counters), $\ell_\init,\ell_\halt\in \Lab$, and $\Inst$ is a mapping assigning to
 each label $\ell\in\Lab\setminus\{\ell_\halt\} $  an instruction for either
\begin{inparaenum}[(i)]
  \item  \emph{increment}: $c_h:= c_h+1$; \texttt{goto} $\ell_r$, or
  \item  \emph{decrement}: \texttt{if} $c_h>0$ \texttt{then} $c_h:= c_h-1$; \texttt{goto} $\ell_s$ \texttt{else goto} $\ell_t$, 
\end{inparaenum}
where $h \in \{1, 2\}$,  $\ell_s\neq \ell_t$, and $\ell_r,\ell_s,\ell_t\in\Lab$.

The machine $M$ induces a transition relation $\longrightarrow$ over configurations of the form
$(\ell, n_1, n_2)$, where $\ell$ is a label of an instruction to be executed and $n_1,n_2\in\Nat$ represent
current values of counters $c_1$ and $c_2$, respectively. A  computation of $M$ is a finite sequence
$C_1\ldots C_k$ of configurations such that $C_i \longrightarrow C_{i+1}$ for all $i\in [1,k-1]$.
The machine $M$ \emph{halts} if there is a computation starting at $(\ell_\init, 0, 0)$ and leading to  configuration
$(\ell_{\halt}, n_1, n_2)$ for some $n_1,n_2\in\Nat$. The halting problem is
to decide whether a given machine $M$ halts. The problem is undecidable~\cite{Minsky67}.
  We adopt the following notation, where $\ell\in\Lab\setminus\{\ell_\halt\} $:
\begin{compactitem}
  \item (i) if $\Inst(\ell)$ is an increment instruction of the form  $c_h:= c_h+1$; \texttt{goto} $\ell_r$, define $c(\ell):= c_h$ and $\Succ(\ell):= \ell_r$;
  (ii) if $\Inst(\ell)$ is a decrement instruction  of the form \texttt{if} $c_h>0$ \texttt{then} $c_h:= c_h-1$; \texttt{goto}  $\ell_r$ \texttt{else goto} $\ell_s$,
  define $c(\ell):= c_h$, $\dec(\ell):= \ell_r$, and $\zero(\ell):= \ell_s$.
\end{compactitem}
\vspace{0.2cm}

We encode the computations of $M$ by using finite words over the pushdown alphabet $\Sigma_\Prop$,
where $\Prop = \Lab\cup \{c_1,c_2\}\cup\{\call,\ret,\intA\}$. For a finite word $\sigma=a_1 \ldots a_n$ over $\Lab\cup\{c_1,c_2\}$, we denote by $\sigma^{R}$ the reverse of $\sigma$, and by $(\call,\sigma)$
(resp., $(\ret,\sigma)$)  the finite word  over $\Sigma_\Prop$ given by $\{a_1,\call\}\ldots \{a_n,\call\}$ (resp., $\{a_1,\ret\}\ldots \{a_n,\ret\}$).
We associate to each $M$-configuration $(\ell,n_1,n_2)$ two distinct encodings: the \emph{call-code} which is the finite word over $\Sigma_\Prop$
given by $(\call,\ell c_1^{n_1} c_2^{n_2})$, and the \emph{ret-code} which is given by  $(\ret,(\ell c_1^{n_1} c_2^{n_2})^{R})$ intuitively corresponding to
the matched-return version of the call-code. A computation $\pi$ of $M$ is then represented by the well-matched word
$(\call,\sigma_\pi)\cdot (\ret,(\sigma_\pi)^{R})$, where  $\sigma_\pi$ is obtained by concatenating the call-codes of the individual configurations along $\pi$.

Formally, let $\Lang_\halt$ be the set of finite words over $\Sigma_\Prop$ of the form
$(\call,\sigma)\cdot (\ret, \sigma^{R})$ (\emph{well-matching requirement}) such that the call part $(\call,\sigma)$ satisfies:
\begin{compactitem}
  \item \emph{Consecution:} $(\call,\sigma)$ is a sequence of call-codes, and for each pair $C\cdot C'$ of adjacent call-codes,
  the associated $M$-configurations, say   $(\ell,n_1,n_2)$ and $(\ell',n'_1,n'_2)$,  satisfy: $\ell\neq \ell_\halt$ and
\begin{inparaenum}[(i)]
  \item  if $\Inst(\ell)$ is an increment instruction and $c(\ell)= c_h$, then $\ell'=\Succ(\ell)$ and $n'_h>0$;
  \item      if   $\Inst(\ell)$ is a decrement instruction and $c(\ell)= c_h$,  then
   \emph{either} $\ell'=\zero(\ell)$ and $n_h=n'_h=0$,
     \emph{or}  $\ell'=\dec(\ell)$ and $n_h>0$.
\end{inparaenum}
  \item \emph{Initialization:}  $\sigma$ has a prefix of the form $\ell_\init \cdot \ell$ for some $\ell\in \Lab$.
    \item \emph{Halting:} $\ell_\halt$ occurs along $\sigma$.
  \item For each pair $C\cdot C'$ of adjacent call-codes in $(\call,\sigma)$ 
  with $C'$ non-halting,
  the relative $M$-configura\-tions  $(\ell,n_1,n_2)$ and $(\ell',n'_1,n'_2)$  satisfy:\footnote{For technical convenience, we do not require that the counters in a configuration having as successor an halting configuration are correctly updated.}
\begin{inparaenum}[(i)]
    \item  \emph{Increment requirement:} if $\Inst(\ell)$ is an increment instruction and $c(\ell)= c_h$, then $n'_h=n_h+1$ and $n'_{3-h}=n_{3-h}$;
        \item \emph{Decrement requirement:} if   $\Inst(\ell)$ is a decrement instruction and $c(\ell)= c_h$,  then $n'_{3-h}=n_{3-h}$, and, if
    $\ell'=\dec(\ell)$, then $n'_h=n_h-1$.
\end{inparaenum}
\end{compactitem}\vspace{0.1cm}

Evidently, $M$ halts iff $\Lang_\halt\neq \emptyset$. We construct in polynomial time a future  \NMTL\ formula $\varphi_M$ over $\Prop$   such that
the set of untimed components $\sigma$ in the finite timed words $(\sigma,\tau)$ satisfying $\varphi_M$ is exactly $\Lang_\halt$.
Hence, Theorem~\ref{theorem:undecidability} directly follows. In the construction of $\varphi_M$, we exploit the future \LTL\ modalities and the abstract next modality $\Next^{\abs}$ which can be expressed in future
\NMTL.

Formally, formula $\varphi_M$ is given by
$
\varphi_M:= \varphi_{\textit{WM}} \vee \varphi_\LTL \vee \varphi_{\textit{Time}}
$
where 
$\varphi_{\textit{WM}}$ is a future \CARET\ formula ensuring the well-matching requirement;
$
\varphi_{\textit{WM}}:= \call \wedge  \Next^{\abs}(\neg\Next^{\Global}\true) \wedge \Always^{\Global}\neg\intA \wedge \neg \Eventually^{\Global}(\ret \wedge \Eventually^{\Global} \call).
$
 The conjunct $\varphi_\LTL$ is a standard future \LTL\ formula ensuring the consecution, initialization, and halting requirements. The definition of $\varphi_\LTL$ is straightforward and we omit the details of the construction.
Finally, we illustrate the construction of the conjunct $\varphi_{\textit{Time}}$ which is a future \MTL\ formula enforcing the increment and decrement requirements by means of time constraints.
Let $w$ be a finite timed word over $\Sigma_\Prop$. By the formulas $\varphi_{\textit{WM}}$ and $\varphi_\LTL$, we can assume that the untime part of $w$ is of the
form $(\call,\sigma)\cdot (\ret, \sigma^{R})$ such that the call part $(\call,\sigma)$ satisfies the consecution, initialization, and halting requirements. Then, formula
$\varphi_{\textit{Time}}$ ensures the following additional requirements:
\begin{compactitem}
  \item \emph{Strict time monotonicity: } the time distance between distinct positions is always greater than zero. This can be expressed by the formula
  $\Always^{\Global}(\neg \StrictEventually^{\Global}_{[0,0]}\top)$.
  \item \emph{1-Time distance between adjacent labels: } the time distance between the $\Lab$-positions of two adjacent $\call$-codes
  (resp., $\ret$-codes) is $1$. This can be expressed as follows:\vspace{0.1cm}

  \hspace{2.0cm}$
  \displaystyle{\bigwedge_{t\in\{\call,\ret\}} \Always^{\Global}\Bigl([t\wedge \bigvee_{\ell\in\Lab}\ell \wedge \StrictEventually^{\Global}(t\wedge \bigvee_{\ell\in\Lab}\ell)] \rightarrow \StrictEventually^{\Global}_{[1,1]}(t\wedge \bigvee_{\ell\in\Lab}\ell) \Bigr)}
   $\vspace{0.1cm}
 \item \emph{Increment and decrement requirements:} fix a $\call$-code $C$ along the call part immediately followed by some non-halting call-code $C'$. Let
 $(\ell,n_1,n_2)$ (resp., $(\ell',n'_1,n'_2)$) be the configuration encoded by $C$ (resp., $C'$), and $c(\ell)=c_h$ (for some $h=1,2$). Note that
 $\ell\neq \ell_\halt$. First, assume that  $\Inst(\ell)$ is an increment instruction.   We need to enforce that
 $n'_h=n_h+1$ and $n'_{3-h}=n_{3-h}$. For this, we first require that:
 (*) for every $\call$-code $C$ with label $\ell$, every  $c_{3-h}$-position has a future  call $c_{3-h}$-position  at (time) distance $1$,
   and every  $c_{h}$-position  has a  future call $c_{h}$-position $j$ at  distance $1$ such that
   $j+1$ is still a call $c_{h}$-position.

 By the strict time monotonicity and the 1-Time distance between adjacent labels, the above requirement~(*) ensures that
 $n'_h\geq n_h+1$ and $n'_{3-h}\geq n_{3-h}$. In order to enforce that $n'_h\leq n_h+1$ and $n'_{3-h}\leq n_{3-h}$, we crucially exploit the return part  $(\ret, \sigma^{R})$
corresponding to the reverse of the call part $(\call,\sigma)$. In particular, along the return part, the reverse of $C'$ is immediately followed by the reverse of $C$. Thus, we additionally require that:
  (**)  for every \emph{non-first} $\ret$-code $R$ which is immediately followed by a $\ret$-code with label $\ell$, each  $c_{3-h}$-position has a future $c_{3-h}$-position  at  distance $1$,
   and each non-first $c_{h}$-position of $R$  has a future  $c_{h}$-position at  distance $1$.

 Requirements~(*) and~(**) can be expressed by the following two formulas.\vspace{0.1cm}

\hspace{1.5cm}  $
  \Always^{\Global}\Bigl((\call\wedge \ell) \rightarrow \StrictAlways^{\Global}_{[0,1]}[(c_{3-h} \rightarrow \StrictEventually^{\Global}_{[1,1]}c_{3-h})\wedge (c_{h} \rightarrow \StrictEventually^{\Global}_{[1,1]}(c_{h}\wedge \Next^{\Global} c_h))]\Bigr)
  $\vspace{0.1cm}

\hspace{0.4cm}$
 \displaystyle{\bigwedge_{\ell'\in \Lab}} 
  \Always^{\Global}\Bigl((\ret\wedge \ell'\wedge \StrictEventually^{\Global}_{[2,2]}\ell) \longrightarrow
 \StrictAlways^{\Global}_{[0,1]}\bigl( [ c_{3-h} \rightarrow \StrictEventually^{\Global}_{[1,1]}c_{3-h}]\wedge
   [(c_{h}\wedge \Next^{\Global} c_h)  \rightarrow \Next^{\Global}\StrictEventually^{\Global}_{[1,1]}c_{h}]\bigr)\Bigr)
$\vspace{0.1cm}

Now, assume that $\Inst(\ell)$ is a decrement instruction.  We need to enforce that
 $n'_{3-h}=n_{3-h}$, and whenever $\ell'=\dec(\ell)$, then $n'_h=n_h-1$. This can be ensured by requirements similar to Requirements~(*) and~(**), and we omit the details.
\end{compactitem}\vspace{0.1cm}

\noindent Note that the unique abstract modality used in the reduction is $\Next^{\abs}$. This concludes the proof of Theorem~\ref{theorem:undecidability}.

\section{Conclusions}
We have introduced two timed linear-time temporal logics for specifying real-time context-free requirements in a pointwise semantics setting: Event-Clock Nested Temporal Logic  (\ECNTL) and Nested Metric Temporal Logic (\NMTL).
 We have shown that while \ECNTL\ is decidable and tractable, \NMTL\ is undecidable even for its future fragment interpreted over finite timed words.
 Moreover, we have established that the \MITLS-like fragment \NMITLS\ of \NMTL\ is decidable and tractable.
 As future research, we shall investigate the decidability properties for the more general fragment of \NMTL\ obtained by disallowing singular intervals.
 Such a fragment represents the \NMTL\ counterpart of Metric Interval Temporal Logic (\MITL), a well-known decidable (and \EXPSPACE-complete) fragment of \MTL\ \cite{AlurFH96}
 which is strictly more expressive than \MITLS\ in the pointwise semantics setting~\cite{RaskinS99}.

\bibliographystyle{eptcs}
\bibliography{bib2}

\end{document}